\documentclass[10pt,reqno,oneside]{amsart}

\usepackage{cancel,bm, amssymb}
\usepackage{color}

\usepackage{enumitem}

\usepackage{color}
\usepackage{graphicx,framed}
\usepackage[colorlinks=true, pdfstartview=FitV, linkcolor=blue, citecolor=blue, urlcolor=blue]{hyperref}
\definecolor{shadecolor}{rgb}{0.95, 0.95, 0.86}

\renewcommand{\d}{{\mathrm d}}
\newcommand{\im}{{\mathrm i}}
\newcommand{\e}{{\mathrm e}}
\newcommand{\z}{\zeta}

\numberwithin{equation}{section}

\newtheorem{theo}{Theorem}[section]

\newtheorem{lem}[theo]{Lemma}
\newtheorem{rem}[theo]{Remark}
\newtheorem{problem}[theo]{Riemann-Hilbert Problem}

\newtheorem{prop}[theo]{Proposition} 
\newtheorem{cor}[theo]{Corollary}

\begin{document}

\title[Asymptotic behavior of a log gas in the bulk scaling limit II.]{On the asymptotic behavior of a log gas in the bulk scaling limit in the presence of a varying external potential II.}

\author{Thomas Bothner}
\address{Department of Mathematics, University of Michigan, 2074 East Hall, 530 Church Street, Ann Arbor, MI 48109-1043, United States}
\email{bothner@umich.edu}

\author{Percy Deift}
\address{Courant Institute of Mathematical Sciences, 251 Mercer St., New York, NY 10012, U.S.A.}
\email{deift@cims.nyu.edu}

\author{Alexander Its}
\address{Department of Mathematical Sciences,
Indiana University-Purdue University Indianapolis,
402 N. Blackford St., Indianapolis, IN 46202, U.S.A.}
\email{itsa@math.iupui.edu}

\author{Igor Krasovsky}
\address{Department of Mathematics, Imperial College, London SW7 2AZ, United Kingdom}
\email{i.krasovsky@imperial.ac.uk}



\thanks{P.Deift acknowledges support of NSF Grant DMS-1300965. A. Its acknowledges support of NSF Grants DMS-1001777 and DMS-1361856, and of SPbGU grant N 11.38.215.2014. I.Krasovsky acknowledges support of the European Community Seventh Framework grant ``Random and Integrable Models in Mathematical Physics".}

\begin{abstract}
In this paper we continue our analysis \cite{BDIK} of the determinant $\det(I-\gamma K_s),\gamma\in(0,1)$ where $K_s$ is the trace class operator acting in $L^2(-1,1)$  with kernel $K_s(\lambda,\mu)=\frac{\sin s(\lambda-\mu)}{\pi(\lambda-\mu)}$. In \cite{BDIK} various key asymptotic results were stated and utilized, but without proof: Here we provide the proofs (see Theorem \ref{theo:1} and Proposition \ref{prop:int} below).
\end{abstract}

\date{\today}
\dedicatory{Dedicated to Albrecht B\"ottcher on the occasion of his $60$th birthday.}
\maketitle

\section{Introduction and statement of results}\label{sec0}
In this paper we consider the asymptotic behavior of the determinant $\det(I-\gamma K_s)$ as $s\rightarrow+\infty$ and $\gamma\in[0,1]$. This problem has received considerable attention over the past $40$ years (see, e.g., \cite{BDIK} for more discussion) and, in particular, the behavior for fixed
\begin{equation*}
	v=-\frac{1}{2}\ln(1-\gamma)\in[0,+\infty]
\end{equation*}
is known: when $v=+\infty$, i.e. $\gamma=1$, we have from \cite{Dy0,K,E,DIKZ}, as $s\rightarrow+\infty$,
\begin{equation}\label{e:1}
	\ln\det(I-K_s)=-\frac{s^2}{2}-\frac{1}{4}\ln s+\ln c_0+\mathcal{O}\left(s^{-1}\right);\ \ \ \ \ \ \ \ln c_0=\frac{1}{12}\ln 2+3\z'(-1),
\end{equation}
which is the classical expansion for the gap probability in the bulk scaling limit for the eigenvalues of Hermitian matrices chosen from the Gaussian Unitary ensemble (GUE). On the other hand, when $v\in[0,+\infty)$ is fixed, we have from \cite{BW,BB}, as $s\rightarrow+\infty$,
\begin{equation}\label{e:2}
	\ln\det(I-\gamma K_s)=-\frac{4v}{\pi}s+\frac{2v^2}{\pi^2}\ln(4s)+2\ln G\left(1+\frac{\im v}{\pi}\right)G\left(1-\frac{\im v}{\pi}\right)+\mathcal{O}\left(s^{-1}\right)
\end{equation}
in terms of Barnes $G$-function, cf. \cite{NIST}. For $\gamma<1$, $-\ln\det(I-\gamma K_s)$ gives the free energy for the so-called log-gas of eigenvalues of GUE Hermitian matrices in the above bulk scaling limit, but now in the presence of an external field $\mathcal{V}(x)=\ln(1-\gamma)$ for $-\frac{2s}{\pi}<x<\frac{2s}{\pi}$, and zero otherwise (see \cite{BDIK}). The mechanism behind the transformation of the exponential decay in \eqref{e:2} to \eqref{e:1} as $v\rightarrow+\infty$ was first analyzed by Dyson in \cite{Dy1}, although on a heuristical level. Nevertheless Dyson's computations served as a guideline in the Riemann-Hilbert based approach chosen in \cite{BDIK} which led to the following result: let $a=a(\kappa)\in(0,1)$ be the unique solution of the equation
\begin{equation}\label{e:0}
	(0,1)\ni\kappa\equiv\frac{v}{s}=\int_a^1\sqrt{\frac{\mu^2-a^2}{1-\mu^2}}\,\d\mu,
\end{equation}
and set
\begin{equation}\label{e:-1}
	V(\kappa)=-\frac{2}{\pi}\Big(E(a)-\big(1-a^2\big)K(a)\Big);\ \ \  \tau(\kappa)=2\im\frac{K(a)}{K(a')};\ \ \ \theta_3(z|\tau)=1+2\sum_{k=1}^{\infty}\e^{\im\pi\tau k^2}\cos(2\pi kz),\ z\in\mathbb{C};
\end{equation}
where $K=K(a)$ and $E=E(a)$ are the standard complete elliptic integrals (cf. \cite{NIST}) with modulus $a=a(\kappa)\in(0,1)$ and complementary modulus $a'=\sqrt{1-a^2}$.
\begin{theo}[\cite{BDIK}, Theorem $1.4$]\label{theo:princ} Given $\delta\in(0,1)$, there exist positive $s_0=s_0(\delta),c=c(\delta)$ and $C_0=C_0(\delta)$ such that
\begin{equation}\label{princ}
	\ln\det(I-\gamma K_s)=-\frac{s^2}{2}\big(1-a^2\big)+vsV+\ln\theta_3(sV|\tau)+A(v)+\int_s^{\infty}M\Big(tV(vt^{-1}),vt^{-1}\Big)\frac{\d t}{t}+J(s,v)
\end{equation}
with
\begin{equation*}
	A(v)=2\ln G\left(1+\frac{\im v}{\pi}\right)G\left(1-\frac{\im v}{\pi}\right)-\frac{v^2}{\pi^2}\left(3+2\ln\left(\frac{\pi}{v}\right)\right),
\end{equation*}
and
\begin{equation*}
	\left|\int_s^{\infty}M\Big(tV(vt^{-1}),vt^{-1}\Big)\frac{\d t}{t}\right|\leq C_0,\ \ \ \ \ \ \ \big|J(s,v)\big|\leq cs^{-\frac{1}{4}}\ln s,
\end{equation*}
for $s\geq s_0$ and $0<v\leq s(1-\delta)$. Here, the function $M(z,\kappa)$ is explicitly stated in $(1.24)$ of \cite{BDIK} as well as in Appendix \ref{appA} ,\eqref{Mid} in terms of Jacobi theta functions.
\end{theo}
The derivation of the bound for $J(s,v)$ in \cite{BDIK} used the following extension of \eqref{e:2} which will be the first result of the current paper.
\begin{theo}\label{theo:1} There exist positive constants $s_0,c_1,c_2$ such that
\begin{equation}\label{e:3}
	\ln\det(I-\gamma K_s)=-\frac{4v}{\pi}s+\frac{2v^2}{\pi^2}\ln(4s)+2\ln G\left(1+\frac{\im v}{\pi}\right)G\left(1-\frac{\im v}{\pi}\right)+r(s,v)
\end{equation}
where $s\geq s_0,0\leq v<s^{\frac{1}{3}}$. The error term $r(s,v)$ is differentiable with respect to $s$ and
\begin{equation}\label{e:4}
	\big|r(s,v)\big|\leq c_1\frac{v}{s}+c_2\frac{v^3}{s}.
\end{equation}
\end{theo}
Note that \eqref{e:3} reduces to \eqref{e:2} if $v$ is fixed. However \eqref{e:3}, \eqref{e:4} is a considerably stronger result than \eqref{e:2} since $v$ is also allowed to grow at a certain rate. As noted in \cite{BDIK}, Theorem \ref{theo:1} is needed, in particular, to make rigorous a very interesting argument of Bohigas and Pato \cite{BP} which interprets the transition \eqref{e:2} to \eqref{e:1} as a transition from a system with Poisson statistics to a system with random matrix theory GUE statistics. We will provide two derivations for Theorem \ref{theo:1}: one based on a nonlinear steepest descent analysis of Riemann-Hilbert Problem $1.3$ in \cite{BDIK}, see also RHP \ref{master} below, but using very different techniques than those employed in \cite{BDIK}. This approach will provide us with expansion \eqref{e:3} involving a slightly weaker estimate for $r(s,v)$ than \eqref{e:4} (see \eqref{p:100} below). Still, the result of this first derivation is sufficient for the arguments used in \cite{BDIK}, Section $5.2$ in the derivation of the bound for $J(s,v)$ in Theorem \ref{theo:princ}. Second, Theorem \ref{theo:1} will also follow by the application of recent results from \cite{CK} using the connection of $\det(I-\gamma K_s)$ to a Toeplitz determinant with Fisher-Hartwig singularities. We confirm in this way \eqref{e:3} with the stated error estimate \eqref{e:4}.\smallskip

Our next result will provide further insight into the integral term involving $M(z,\kappa)$.
\begin{prop}\label{prop:int} Let $a_0(u)=\int_0^1M(x,u)\d x$ denote an average of $M(x,\kappa)$ over the ``fast" variable $x$. Given $\delta\in(0,1)$, we have
\begin{equation*}
	\int_s^{\infty}M\Big(tV(vt^{-1}),vt^{-1}\Big)\frac{\d t}{t}=\int_0^{\kappa}a_0(u)\frac{\d u}{u}+\mathcal{O}\left(s^{-1}\right)
\end{equation*}
uniformly for $s\geq s_0>0$ and $0<v\leq s(1-\delta)$.
\end{prop}
It is worth noticing that Proposition \ref{prop:int} implies that the integral term in the right hand side of \eqref{princ} depends, up to $\mathcal{O}(s^{-1})$, on the ``slow" variable $\kappa=\frac{v}{s}$ only. Hence, Theorem \ref{theo:princ} captures in an explicit way all principal features of the leading asymptotic behavior of the determinant $\det(I-\gamma K_s)$ in the bulk scaling limit.
\begin{rem} We conjecture that
\begin{equation*}
	a_0(\kappa)=0,
\end{equation*}
so that the integral term in \eqref{princ} does not contribute at all to the leading order asymptotics of the sine-kernel determinant. This conjecture has been checked to leading order in the limit $\kappa\downarrow 0$. We have
\begin{equation*}
	M(x,\kappa)=-\frac{\kappa}{6\pi}\cos(2\pi x)+\mathcal{O}\left(\kappa^2\right),\ \ \kappa\downarrow 0,
\end{equation*}
and thus $\int_0^1M(x,\kappa)\d x=\mathcal{O}\left(\kappa^2\right)$.
\end{rem}
\subsection{Outline of the paper} We complete the introduction with a short outline for the remainder. As mentioned above, the leading terms in Theorem \ref{theo:1} will be derived in two ways: First, in Section \ref{sec:2}, we apply Riemann-Hilbert techniques related to the integrable structure of $K_s(\lambda,\mu)$, cf. \cite{IIKS}. The $s$ derivative of $\ln\det(I-\gamma K_s)$ is expressible in terms of the solution of RHP \ref{master} below which allows us to compute its asymptotic by application of nonlinear steepest descent techniques \cite{DZ}. Integrating the asymptotic series indefinitely with respect to $s$ and comparing the result to \eqref{e:2} we obtain \eqref{e:3} with a slightly weaker error estimate for $r(s,v)$. Second, in Section \ref{sec:3}, we use the representation of $\det(I-\gamma K_s)$ as a limit of a Toeplitz determinant with Fisher-Hartwig singularities. We strengthen the results obtained in \cite{CK} and derive Theorem \ref{theo:1} with the stated error estimate \eqref{e:4}. In Section \ref{sec:4} we provide a proof of Proposition \ref{prop:int} which uses a Fourier series representation of $M(z,\kappa)$. Several theta function expressions which are used in the definition of $M(z,\kappa)$ are summarized in Appendix \ref{appA}.

\section{Extending \eqref{e:2} by integrable operator techniques}\label{sec:2}

We first recall the central Riemann-Hilbert problem related to the integrable structure of the sine kernel determinant $\det(I-\gamma K_s)$, compare \cite{IIKS}.
\begin{problem}\label{master} Determine $Y(\lambda)=Y(\lambda;s,\gamma)\in\mathbb{C}^{2\times 2}$ such that
\begin{enumerate}
	\item $Y(\lambda)$ is analytic for $\lambda\in\mathbb{C}\backslash[-1,1]$ with square integrable limiting values 
	\begin{equation*}
		Y_{\pm}(\lambda)=\lim_{\varepsilon\downarrow 0}Y(\lambda\pm\im\varepsilon),\ \ \lambda\in(-1,1).
	\end{equation*}
	\item The boundary values $Y_{\pm}(\lambda)$ are related by the jump condition
	\begin{equation*}
		Y_+(\lambda)=Y_-(\lambda)\begin{pmatrix} 1-\gamma & \gamma\e^{2\im\lambda s}\smallskip\\ -\gamma\e^{-2\im\lambda s} & 1+\gamma \end{pmatrix},\ \ \lambda\in(-1,1).
	\end{equation*}
	\item Near the endpoints $\lambda=\pm 1$, we have
	\begin{equation*}
		Y(\lambda)=\check{Y}(\lambda)\left[I+\frac{\gamma}{2\pi\im}\begin{pmatrix}-1 & 1\\ -1 & 1 \end{pmatrix}\ln\left(\frac{\lambda-1}{\lambda+1}\right)\right]\e^{-\im s\lambda\sigma_3},\ \ \lambda\rightarrow\pm 1;\hspace{0.7cm}\sigma_3=\bigl(\begin{smallmatrix} 1 & 0\smallskip\\ 0 & -1 \end{smallmatrix}\bigr),
	\end{equation*}
	where $\check{Y}(\lambda)$ is analytic at $\lambda=\pm 1$ and $\ln$ denotes the principal branch for the logarithm.
	\item Near $\lambda=\infty$,
	\begin{equation*}
		Y(\lambda)=I+Y_1\lambda^{-1}+\mathcal{O}\left(\lambda^{-2}\right),\ \ \lambda\rightarrow\infty;\ \ \ \ \ \ Y_1=\big(Y_1^{jk}\big).
	\end{equation*}
\end{enumerate}
\end{problem}
Our goal is to solve this problem asymptotically for sufficiently large $s\geq s_0$ and $0\leq v<s^{\frac{1}{3}}$. This will be achieved by an application of the nonlinear steepest descent method of Deift and Zhou \cite{DZ} to RHP \ref{master}. Our approach is somewhat similar to \cite{BI}.
\subsection{Matrix factorization and opening of lens} We make use of a factorization of the jump matrix $G_Y(\lambda)$ occurring in RHP \ref{master},
\begin{equation*}
	G_Y(\lambda)=\begin{pmatrix} 1-\gamma & \gamma\e^{2\im\lambda s}\smallskip\\ -\gamma\e^{-2\im\lambda s} & 1+\gamma \end{pmatrix}=\begin{pmatrix} 1 & 0\\ -\gamma\e^{2s(\kappa-\im\lambda)} & 1 \end{pmatrix}\e^{-2\kappa s\sigma_3}\begin{pmatrix} 1 & \gamma\e^{2s(\kappa+\im\lambda)}\\ 0 & 1 \end{pmatrix}\equiv S_L(\lambda)S_DS_U(\lambda).
\end{equation*}
Since $S_L(\lambda)$ and $S_U(\lambda)$ clearly admit analytical continuations to the lower and upper $\lambda$-planes respectively, we can perform a first transformation of the initial RHP \ref{master}: introduce
\begin{equation*}
	S(\lambda)=Y(\lambda)\begin{cases} S_U^{-1}(\lambda),&\lambda\in\Omega_1\\ S_L(\lambda),&\lambda\in\Omega_2\\ I,& \textnormal{else}\end{cases}
\end{equation*}
where the domains $\Omega_j\subset\mathbb{C}$ are sketched in Figure \ref{figure1} below. This leads to the following problem:
\begin{figure}[tbh]
\begin{center}
\resizebox{0.325\textwidth}{!}{\includegraphics{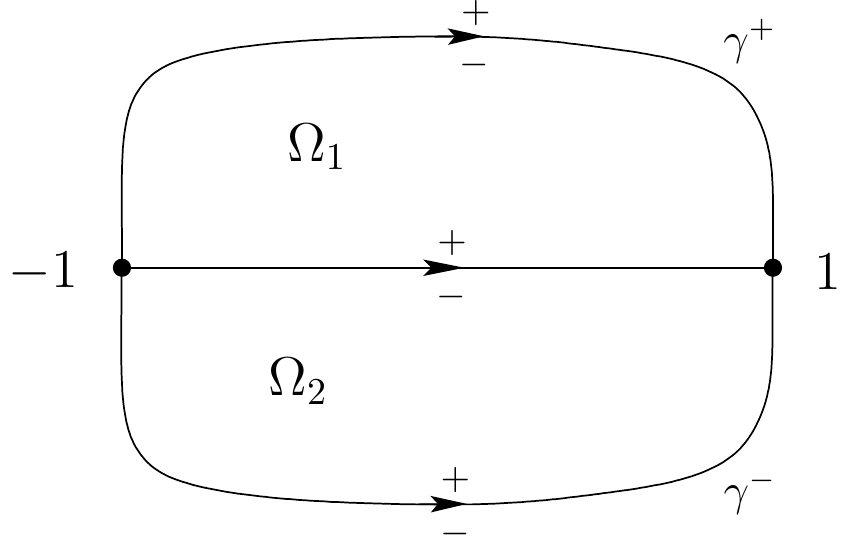}}
\caption{The oriented jump contour $\Sigma_S$ for the function $S(\lambda)$ in the complex $\lambda$-plane.}
\label{figure1}
\end{center}
\end{figure}
\begin{problem}\label{open} Determine $S(\lambda)=S(\lambda;s,\gamma)\in\mathbb{C}^{2\times 2}$ such that
\begin{enumerate}
	\item $S(\lambda)$ is analytic for $\lambda\in\mathbb{C}\backslash\Sigma_S$ where the jump contour $\Sigma_S$ is shown in Figure \ref{figure1}.
	\item The limiting values $S_{\pm}(\lambda),\lambda\in\Sigma_S$ are square integrable and satisfy the jump conditions
	\begin{equation*}
		S_+(\lambda)=S_-(\lambda)S_U(\lambda),\ \ \lambda\in\gamma^+;\ \ \ S_+(\lambda)=S_-(\lambda)S_L(\lambda),\ \ \lambda\in\gamma^-;\ \ \ S_+(\lambda)=S_-(\lambda)S_D,\ \ \lambda\in(-1,1).
	\end{equation*}
	\item Near $\lambda=\pm 1$,
	\begin{equation*}
		S(\lambda)=\check{Y}(\lambda)\left[I+\frac{\gamma}{2\pi\im}\begin{pmatrix} -1 & 1\\ -1 & 1 \end{pmatrix}\ln\left(\frac{\lambda-1}{\lambda+1}\right)\right]\e^{-\im s\lambda\sigma_3}\begin{cases} S_U^{-1}(\lambda),& \lambda\in\Omega_1\\ S_L(\lambda),&\lambda\in\Omega_2\\ I,&\textnormal{else}\end{cases}.
	\end{equation*}
	\item As $\lambda\rightarrow\infty$, we have $S(\lambda)\rightarrow I$.
\end{enumerate}
\end{problem}
If $v\in(0,\infty)$ were kept fixed throughout we would already know at this point that the major contribution to the asymptotic solution of RHP \ref{open} arises from the segment $(-1,1)$ and two small neighborhoods of the endpoints $\lambda=\pm 1$. Indeed we have
\begin{equation*}
	\big|\e^{2s(\kappa\pm\im\lambda)}\big| = \e^{2v}\e^{\mp 2s\Im\lambda},
\end{equation*}
and thus exponentially fast decay to zero on $\gamma^+\backslash D(\pm 1,r)$ in the upper half-plane and $\gamma^-\backslash D(\pm 1,r)$ in the lower half-plane. Here, $D(\pm 1,r)=\{\lambda\in\mathbb{C}:\ |\lambda\mp 1|<r\}$. We will return to this observation after the next subsection.
\subsection{Local Riemann-Hilbert analysis} We require three explicit model functions among which
\begin{equation}\label{p:1}
	P^{(\infty)}(\lambda)=\left(\frac{\lambda-1}{\lambda+1}\right)^{\frac{\im v}{\pi}\sigma_3},\ \ \ \lambda\in\mathbb{C}\backslash[-1,1]:\hspace{0.5cm}P^{(\infty)}(\lambda)\rightarrow I,\ \ \lambda\rightarrow+\infty,
\end{equation}
will serve as the outer parametrix. This function reproduces exactly the jump behavior of $S(\lambda)$ on the segment $(-1,1)$. For the local parametrices near the endpoints we follow and adapt the constructions in \cite{BI}, Sections $9$ and $10$. First we define on the punctured plane $\z\in\mathbb{C}\backslash\{0\}$ with $-\pi<\textnormal{arg}\,\z\leq\pi$,
\begin{equation}\label{p:2}
	P(\z)=\begin{pmatrix} U(-\nu;\e^{-\im\frac{\pi}{2}}\z) & -U(1+\nu;\e^{\im\frac{\pi}{2}}\z)\im\e^{\im\pi\nu}\frac{\Gamma(1+\nu)}{\Gamma(-\nu)}\smallskip\\ U(1-\nu;\e^{-\im\frac{\pi}{2}}\z)\im\e^{\im\pi\nu}\frac{\Gamma(1-\nu)}{\Gamma(\nu)} & U(\nu;\e^{\im\frac{\pi}{2}}\z)\e^{2\pi\im\nu} \end{pmatrix} \e^{\frac{\im}{2}\z\sigma_3}\e^{\im\frac{\pi}{2}(\frac{1}{2}-\nu)\sigma_3}\mathcal{S}(\z)
\end{equation}
with
\begin{equation*}
	\nu=\frac{\im v}{\pi};\ \ \ \ \ \mathcal{S}(\z)=\begin{cases}\Bigl(\begin{smallmatrix} 1 & \frac{2\pi\im\e^{\im\pi\nu}}{\Gamma(\nu)\Gamma(1-\nu)} \\ 0 & 1 \end{smallmatrix}\Bigr)\Bigl(\begin{smallmatrix} \e^{\im\frac{\pi}{2}\nu} & 0 \\ 0 & \e^{-\im\frac{3\pi}{2}\nu} \end{smallmatrix}\Bigr),&\textnormal{arg}\,\z\in(\frac{\pi}{2},\pi)\smallskip\\ \Bigl(\begin{smallmatrix} 1 &0\\ \frac{2\pi\im\,\e^{-3\pi\im\nu}}{\Gamma(\nu)\Gamma(1-\nu)} & 1 \end{smallmatrix}\Bigr)\Bigl(\begin{smallmatrix}\e^{\im\frac{\pi}{2}\nu} & 0 \\ 0 & \e^{-\im\frac{3\pi}{2}\nu} \end{smallmatrix}\Bigr),&\textnormal{arg}\,\z\in(-\pi,-\frac{\pi}{2})\smallskip\\ \Bigl(\begin{smallmatrix} \e^{\im\frac{\pi}{2}\nu} & 0 \\ 0 & \e^{-\im\frac{3\pi}{2}\nu} \end{smallmatrix}\Bigr),&\textnormal{arg}\,\z\in(-\frac{\pi}{2},\frac{\pi}{2})\end{cases}
\end{equation*}
and $U(a;\z)=U(a,1;\z)\sim\z^{-a}$ as $\z\rightarrow\infty,\textnormal{arg}\,\z\in(-\frac{3\pi}{2},\frac{3\pi}{2})$, is the confluent hypergeometric function, cf. \cite{NIST}. Using standard asymptotic and monodromy properties of $U(a;\z)$, see \cite{NIST}, one can check that the model function $P(\z)$ is a solution to the ``bare" model problem below (cf. \cite{DIII}).
\begin{problem}\label{bareconf} The function $P(\z)$ defined in \eqref{p:2} has the following properties:
\begin{enumerate}
	\item $P(\z)$ is analytic for $\z\in\mathbb{C}\backslash\{\textnormal{arg}\,\z=-\pi,-\frac{\pi}{2},\frac{\pi}{2}\}$ and the orientation of the three rays is fixed as indicated in Figure \ref{figure2} below.
	\item $P(\z)$ has $\z$-independent jumps,
	\begin{equation*}
		P_+(\z)=P_-(\z)\bigl(\begin{smallmatrix} 1 & \gamma\e^{2\kappa s}\\ 0 & 1 \end{smallmatrix}\bigr),\ \ \textnormal{arg}\,\z=\frac{\pi}{2};\ \ \ \ \ P_+(\z)=P_-(\z)\bigl(\begin{smallmatrix}1 & 0\\ -\gamma\e^{2\kappa s} & 1 \end{smallmatrix}\bigr),\ \ \textnormal{arg}\,\z=-\frac{\pi}{2},
	\end{equation*}
	and,
	\begin{equation*}
		P_+(\z)=P_-(\z)\e^{-2\kappa s\sigma_3},\ \ \textnormal{arg}\,\z=-\pi.
	\end{equation*}
	\begin{figure}[tbh]
\begin{center}
\resizebox{0.27\textwidth}{!}{\includegraphics{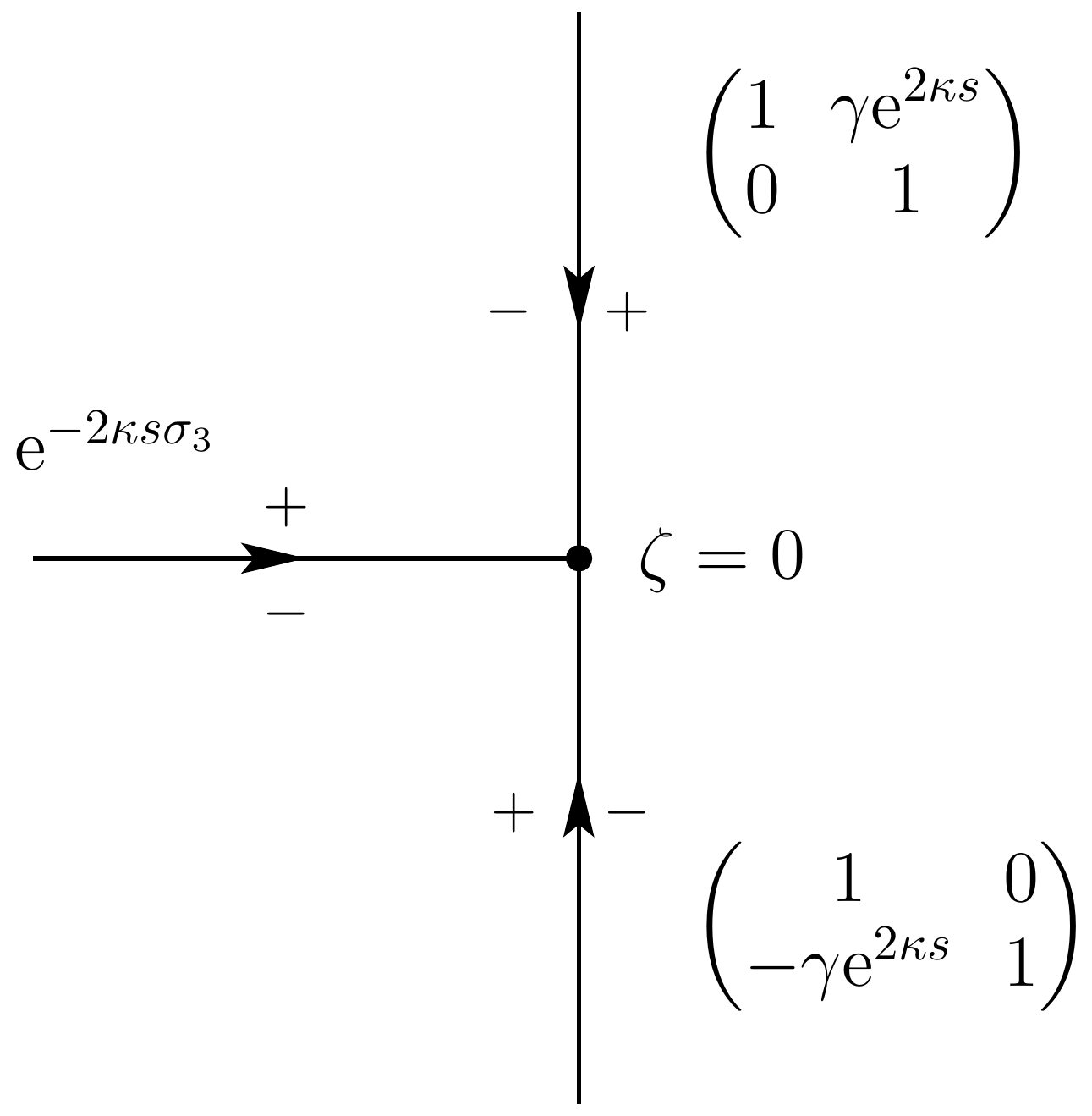}}
\caption{Jump behavior of $P(\z)$ in the complex $\z$-plane.}
\label{figure2}
\end{center}
\end{figure}
	\item Near $\z=0$, with $-\pi<\textnormal{arg}\,\z\leq\pi$,
	\begin{equation}\label{p:3}
		P(\z)=\check{P}(\z)\left[I+\frac{\gamma}{2\pi\im}\begin{pmatrix} -1 & 1\\ -1 & 1 \end{pmatrix}\ln\z\right]\begin{cases}\bigl(\begin{smallmatrix} 1 & -\gamma\e^{2\kappa s}\\ 0 & 1 \end{smallmatrix}\bigr),&\textnormal{arg}\,\z\in(\frac{\pi}{2},\pi)\smallskip\\ \bigl(\begin{smallmatrix} 1 & 0\\ -\gamma\e^{2\kappa s} & 1 \end{smallmatrix}\bigr),&\textnormal{arg}\,\z\in(-\pi,-\frac{\pi}{2})\smallskip\\ I,&\textnormal{arg}\,\z\in(-\frac{\pi}{2},\frac{\pi}{2})\end{cases}
	\end{equation}
	where $\check{P}(\z)$ is analytic at $\z=0$. 
	\item As $\z\rightarrow\infty$,
	\begin{align*}
		P(\z)\sim&\,\left[I+\sum_{k=1}^{\infty}\begin{pmatrix} \big((-\nu)_k\big)^2\e^{-\im\frac{\pi}{2}k} & \im\big((1+\nu)_{k-1}\big)^2k\,\e^{\im\frac{\pi}{2}k-\im\pi\nu}\frac{\Gamma(1+\nu)}{\Gamma(-\nu)}\\ -\im\big((1-\nu)_{k-1}\big)^2k\,\e^{-\im\frac{\pi}{2}k+\im\pi\nu}\frac{\Gamma(1-\nu)}{\Gamma(\nu)} & \big((\nu)_k\big)^2\e^{\im\frac{\pi}{2}k} \end{pmatrix}\frac{\z^{-k}}{k!}\right]\\
		&\,\times\,\z^{\nu\sigma_3}\e^{\frac{\im}{2}\z\sigma_3}\e^{\im\frac{\pi}{2}(\frac{1}{2}-\nu)\sigma_3};
	\end{align*}
	where $(a)_k=a(a+k)(a+2)\cdot\ldots\cdot(a+k-1)$ is Pochhammer's symbol. Observe that in our situation $\gamma\in[0,1]$ and hence we have $\nu\in\im\mathbb{R}$.
\end{enumerate}
\end{problem}
In terms of $P(\z)$ the actual parametrix near $\lambda=1$ is then defined as follows,
\begin{equation}\label{p:4}
	P^{(1)}(\lambda)=\big(2s(\lambda+1)\big)^{-\nu\sigma_3}\e^{-\im\frac{\pi}{2}(\frac{1}{2}-\nu)\sigma_3}\e^{\im s\sigma_3}P\big(\z(\lambda)\big)\e^{-\frac{\im}{2}(\z(\lambda)+2s)\sigma_3},\ \ |\lambda-1|<\frac{1}{4}
\end{equation}
where $\z=\z(\lambda)=2s(\lambda-1)$ denotes the locally conformal change of variables from the $\lambda$- to the $\z$-plane. Using the properties listed in RHP \ref{bareconf} we immediately establish that the initial function $S(\lambda)$ in RHP \ref{open} and \eqref{p:4} are related by an analytic left multiplier
\begin{equation*}
	S(\lambda)=N_1(\lambda)P^{(1)}(\lambda),\ \ \ |\lambda-1|<\frac{1}{4}.
\end{equation*}
Moreover, as $s\rightarrow\infty$, and hence $|\z|\rightarrow\infty$, with $0<r_1\leq|\lambda-1|\leq r_2<\frac{1}{4}$, we have asymptotic matching of the model functions $P^{(\infty)}(\lambda)$ and $P^{(1)}(\lambda)$: with $\beta(\lambda)=\big(2s(\lambda+1)\big)^{\nu}$,
\begin{align}
	P^{(1)}(\lambda)\sim&\,\Bigg[I+\sum_{k=1}^{\infty}\begin{pmatrix}\big((-\nu)_k\big)^2\e^{-\im\frac{\pi}{2}k} & \big((1+\nu)_{k-1}\big)^2k\,\e^{\im\frac{\pi}{2}k}\beta^{-2}(\lambda)\e^{2\im s}\frac{\Gamma(1+\nu)}{\Gamma(-\nu)}\\ \big((1-\nu)_{k-1}\big)^2k\,\e^{-\im\frac{\pi}{2}k}\beta^2(\lambda)\e^{-2\im s}\frac{\Gamma(1-\nu)}{\Gamma(\nu)} & \big((\nu)_k\big)^2\e^{\im\frac{\pi}{2}k} \end{pmatrix}\frac{\z^{-k}}{k!}\Bigg]\nonumber\\
	&\,\times P^{(\infty)}(\lambda).\label{p:5}
\end{align}
Near the remaining endpoint $\lambda=-1$ we can either carry out explicitly a similar construction or simply use symmetry: for $|\lambda+1|<\frac{1}{4}$, introduce the parametrix as
\begin{equation}\label{p:6}
	P^{(-1)}(\lambda)=\sigma_1P^{(1)}(-\lambda)\sigma_1,\ \ \ \ \ \ \sigma_1=\bigl(\begin{smallmatrix} 0 & 1\\ 1 & 0 \end{smallmatrix}\bigr)
\end{equation}
and obtain at once
\begin{equation*}
	S(\lambda)=N_2(\lambda)P^{(-1)}(\lambda),\ \ \ |\lambda+1|<\frac{1}{4},
\end{equation*}
as well as for $s\rightarrow\infty$ with $0<r_1\leq|\lambda+1|\leq r_2<\frac{1}{4}$,
\begin{align}
	P^{(-1)}(\lambda)\sim&\,\Bigg[I+\sum_{k=1}^{\infty}\begin{pmatrix}\big((\nu)_k\big)^2\e^{\im\frac{\pi}{2}k} &\big((1-\nu)_{k-1}\big)^2k\,\e^{-\im\frac{\pi}{2}k}\beta^2(-\lambda)\e^{-2\im s}\frac{\Gamma(1-\nu)}{\Gamma(\nu)}\\  \big((1+\nu)_{k-1}\big)^2k\,\e^{\im\frac{\pi}{2}k}\beta^{-2}(-\lambda)\e^{2\im s}\frac{\Gamma(1+\nu)}{\Gamma(-\nu)} & \big((-\nu)_k\big)^2\e^{-\im\frac{\pi}{2}k} \end{pmatrix}\nonumber\\
	&\,\times\frac{\z(-\lambda)^{-k}}{k!}\Bigg]P^{(\infty)}(\lambda).\label{p:7}
\end{align}
This completes the construction of the explicit model functions $P^{(\infty)}(\lambda),P^{(1)}(\lambda)$ and $P^{(-1)}(\lambda)$.
\subsection{Ratio transformation and small norm estimations} We use \eqref{p:1}, \eqref{p:4}, \eqref{p:6} and define in this step
\begin{equation}\label{p:8}
	R(\lambda)=S(\lambda)\begin{cases}\big(P^{(1)}(\lambda)\big)^{-1},&|\lambda-1|<r\\ \big(P^{(-1)}(\lambda)\big)^{-1},&|\lambda+1|<r\\ \big(P^{(\infty)}(\lambda)\big)^{-1},&|\lambda\mp 1|>r\end{cases}
\end{equation}
with $0<r<\frac{1}{4}$. Recalling the results of the previous subsection we are lead to the problem below.
\begin{figure}[tbh]
\begin{center}
\resizebox{0.35\textwidth}{!}{\includegraphics{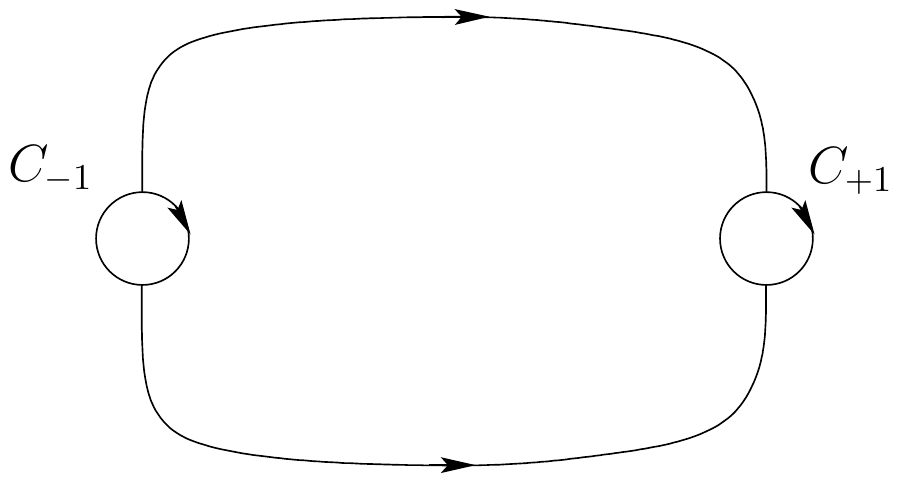}}
\caption{The oriented jump contour $\Sigma_R$ in the ratio RHP \ref{ratioRHP}.}
\label{figure3}
\end{center}
\end{figure}

\begin{problem}\label{ratioRHP} The ratio function $R(\lambda)=R(\lambda;s,\gamma)\in\mathbb{C}^{2\times 2}$ is determined by the following Riemann-Hilbert problem $(\Sigma_R,G_R(\cdot;s,\gamma)):$
\begin{enumerate}
	\item $R(\lambda)$ is analytic for $\lambda\in\mathbb{C}\backslash\Sigma_R$ where the oriented jump contour $\Sigma_R$ is depicted in Figure \ref{figure3}.
	\item The square integrable boundary values $R_{\pm}(\lambda),\lambda\in\Sigma_R$ are related by the jump conditions
	\begin{eqnarray*}
		R_+(\lambda)&=&R_-(\lambda)\left[I+j_{12}(\lambda;s,\gamma)\bigl(\begin{smallmatrix} 0 & 1\\ 0 & 0 \end{smallmatrix}\bigr)\right],\ \ \lambda\in(\Sigma_R\backslash C_{\pm 1})\cap\{\Im\lambda>0\};\\
		R_+(\lambda)&=&R_-(\lambda)\left[I+j_{21}(\lambda;s,\gamma)\bigl(\begin{smallmatrix} 0 & 0\\ 1 & 0 \end{smallmatrix}\bigr)\right],\ \ \lambda\in(\Sigma_R\backslash C_{\pm 1})\cap\{\Im\lambda<0\}
	\end{eqnarray*}
	with
	\begin{equation*}
		j_{12}(\lambda;s,\gamma)=\gamma\e^{2s(\kappa+\im\lambda)}\left(\frac{\lambda-1}{\lambda+1}\right)^{\frac{2\im v}{\pi}},\ \ \ \ \ j_{21}(\lambda)=-\gamma\e^{2s(\kappa-\im\lambda)}\left(\frac{\lambda-1}{\lambda+1}\right)^{-\frac{2\im v}{\pi}}.
	\end{equation*}
	Also, along the clockwise oriented circle boundaries $C_{\pm 1}$,
	\begin{equation*}
		R_+(\lambda)=R_-(\lambda)P^{(1)}(\lambda)\big(P^{(\infty)}(\lambda)\big)^{-1},\ \ \lambda\in C_{+1};\ \ \ \ \ R_+(\lambda)=R_-(\lambda)P^{(-1)}(\lambda)\big(P^{(\infty)}(\lambda)\big)^{-1},\ \ \lambda\in C_{-1}.
	\end{equation*}
	\item As $\lambda\rightarrow\infty$, we have that $R(\lambda)\rightarrow I$.
\end{enumerate}
\end{problem}
On the way to small norm estimations for the jump matrix $G_R(\lambda;s,\gamma),\lambda\in\Sigma_R$ we first note that
\begin{equation*}
	\big|j_{12}(\pm 1+\im r;s,\gamma)\big|\leq \exp\big[2v-v\alpha(r)-2rs\big];\ \ \ \ \ \ \ \big|j_{21}(\pm 1-\im r;s,\gamma)\big|\leq \exp\big[2v-v\alpha(r)-2rs\big]
\end{equation*}
with
\begin{equation*}
	\alpha(r)=1-\frac{2}{\pi}\arctan\left(\frac{r}{2}\right)>1-\frac{2}{\pi}\arctan\left(\frac{1}{8}\right)>0.8.
\end{equation*}
In order to obtain an estimate in the scaling region
\begin{equation*}
	s\geq s_0,\ \ 0\leq v<s^{\frac{1}{3}}
\end{equation*}
we have to choose a contracting radius in \eqref{p:8}. In fact we will work with $r=r(s)=s^{-\frac{1}{3}}$ which leads to
\begin{prop}\label{DZ:1} There exists $s_0>0$ such that
\begin{equation*}
	\|G_R(\cdot;s,\gamma)-I\|_{L^2\cap L^{\infty}(\Sigma_R\backslash C_{\pm 1})}\leq \e^{-s^{\frac{2}{3}}},\ \ \ \ \forall\,s\geq s_0,\ 0\leq v<s^{\frac{1}{3}}.
\end{equation*}
\end{prop}
For $C_{\pm 1}$ we go back to \eqref{p:5} and \eqref{p:7}, first for $\lambda\in C_{+1}$,
\begin{align*}
	\e^{\nu\ln|2s(\lambda+1)|\sigma_3}&\e^{-\frac{\im}{2}\textnormal{arg}(\frac{\Gamma(\nu)}{\Gamma(-\nu)})\sigma_3}\e^{-\im s\sigma_3}\Big(G_R(\lambda;s,\gamma)-I\Big)\e^{\im s\sigma_3}\e^{\frac{\im}{2}\textnormal{arg}(\frac{\Gamma(\nu)}{\Gamma(-\nu)})\sigma_3}\e^{-\nu\ln|2s(\lambda+1)|\sigma_3}\\
	&\sim\sum_{k=1}^{\infty}\begin{pmatrix}\big((-\nu)_k\big)^2 & \nu\big((1+\nu)_{k-1}\big)^2k\,\e^{-2\im\nu\textnormal{arg}(\lambda+1)}\\ -\nu\big((1-\nu)_{k-1}\big)^2k\,\e^{2\im\nu\textnormal{arg}(\lambda+1)} & \big((\nu)_k\big)^2\end{pmatrix}\e^{-\im\frac{\pi}{2}k\sigma_3}\frac{\z^{-k}}{k!}.
\end{align*}
Thus,
\begin{equation*}
	\|G_R(\lambda;s,\gamma)-I\|\leq c\left\{\frac{v}{sr}\left\|\begin{pmatrix} v & \e^{\frac{2v}{\pi}\textnormal{arg}(\lambda+1)}\\ \e^{-\frac{2v}{\pi}\textnormal{arg}(\lambda+1)} & v\end{pmatrix}\right\|+\mathcal{E}(\lambda;s,\gamma)\right\},\ \ \lambda\in C_{+1}
\end{equation*}
and estimations for the error term $\mathcal{E}(\lambda;s,\gamma)$ follow from known error-term estimates for the confluent hypergeometric functions which are, for instance, given in \cite{NIST}: there exists $s_0>0$ and constants $c_j>0$ such that
\begin{equation*}
	\|\mathcal{E}(\lambda;s,\gamma)\|\leq c_1\e^{c_2v-c_3sr}\leq c_1\e^{-c_4s^{\frac{2}{3}}},\ \ \ \ \forall\, s\geq s_0,\ \ 0\leq v<s^{\frac{1}{3}}.
\end{equation*}
Combining this estimate with
\begin{equation*}
	\exp\left[\frac{2v}{\pi}\big|\textnormal{arg}(\lambda+1)\big|\right]\leq\e^{\frac{2v}{\pi}\delta(r)},\ \ \ \delta(r)=\arctan\left(\frac{r}{2}\right),\ \ \lambda\in C_{+1},
\end{equation*}
we obtain (with similar methods for the estimates on $C_{-1}$)
\begin{prop}\label{DZ:2} For every $0<\epsilon<\frac{1}{3}$ there exists $s_0>0$ and $c>0$ such that
\begin{equation*}
	\|G_R(\cdot;s,\gamma)-I\|_{L^2\cap L^{\infty}(C_{\pm 1})}\leq c\,v^2 s^{-\frac{2}{3}},\ \ \ \forall\,s\geq s_0,\ \ 0\leq v<s^{\frac{1}{3}-\epsilon}.
\end{equation*}
\end{prop}
Combining Propositions \ref{DZ:1} and \ref{DZ:2} we obtain by general theory \cite{DZ},
\begin{prop}\label{DZ:3} For every $0<\epsilon<\frac{1}{3}$ there exists $s_0>0$ and $c>0$ such that the ratio RHP \ref{ratioRHP} is uniquely solvable for all $s\geq s_0$ and $0\leq v< s^{\frac{1}{3}-\epsilon}$. The solution can be computed iteratively through the integral equation
\begin{equation*}
	R(\lambda)=I+\frac{1}{2\pi\im}\int_{\Sigma_R}R_-(w)\big(G_R(w)-I\big)\frac{\d w}{w-\lambda},\ \ \lambda\in\mathbb{C}\backslash\Sigma_R,
\end{equation*}
where we use that
\begin{equation*}
	\|R_-(\cdot;s,\gamma)-I\|_{L^2(\Sigma_R)}\leq c\,v^2s^{-\frac{2}{3}},\ \ \forall\,s\geq s_0,\ \ 0\leq v<s^{\frac{1}{3}-\epsilon}.
\end{equation*}
\end{prop}
This concludes the nonlinear steepest descent analysis of RHP \ref{master}.
\subsection{Proof of expansion \eqref{e:3}}
In order to obtain the statement in \eqref{e:3} we make use of the identity (see, e.g. \cite{BDIK}, equation $(2.1)$),
\begin{equation}\label{p:9}
	\frac{\partial}{\partial s}\ln\det(I-\gamma K_s)=-2\im Y_1^{11}
\end{equation}
where $Y_1$ appeared in RHP \ref{master} and the derivative is taken with $\gamma$ fixed. Tracing back the transformations $Y(\lambda)\mapsto S(\lambda)\mapsto R(\lambda)$, we have 
\begin{equation*}
	Y_1=-\frac{2\im v}{\pi}\sigma_3+\frac{\im}{2\pi}\int_{\Sigma_R}R_-(w)\big(G_R(w)-I\big)\d w
\end{equation*}
and with the help of an explicit residue computation as well as Proposition \ref{DZ:3},
\begin{eqnarray*}
	Y_1^{11}&=&-\frac{2\im v}{\pi}+\frac{\im v^2}{s\pi^2}+\frac{\im}{2\pi}\oint_{C_{\pm 1}}\Big[\big(R_-(w)-I\big)\big(G_R(w)-I\big)\Big]_{11}\d w\\
	&=&-\frac{2\im v}{\pi}+\frac{\im v^2}{s\pi^2}+\frac{\im v^2}{4\pi^2s^2}\sin\big(\phi(s,v)\big)+\mathcal{O}\left(v^3s^{-2}\right),\ \ \phi(s,v)=4s-\frac{4v}{\pi}\ln(4s)+2\textnormal{arg}\,\left(\frac{\Gamma(\nu)}{\Gamma(-\nu)}\right).
\end{eqnarray*}
With this we go back to \eqref{p:9} and perform an indefinite integration with respect to $s$,
\begin{equation}\label{p:10}
	\ln\det(I-\gamma K_s)=-\frac{4v}{\pi}s+\frac{2v^2}{\pi^2}\ln s+C(v)+r(s,v)
\end{equation}
and the error term $r(s,v)$ is differentiable with respect to $s$ and for any $0<\epsilon<\frac{1}{3}$ there exist $s_0>0,c>0$ such that
\begin{equation}\label{p:100}
	\big|r(s,v)|\leq c\frac{v^3}{s},\ \ \ \ \ \forall\,s\geq s_0,\ \ 0\leq v<s^{\frac{1}{3}-\epsilon}.
\end{equation}
The term $C(v)$ appearing in \eqref{p:10} is $s$-independent and can therefore be determined by comparison with \eqref{e:2}, i.e. we have
\begin{equation*}
	C(v)=\frac{2v^2}{\pi^2}\ln 4+2\ln G\left(1+\frac{\im v}{\pi}\right)G\left(1-\frac{\im v}{\pi}\right),
\end{equation*}
and this completes the section.
\begin{rem} Estimate \eqref{p:10}, \eqref{p:100} is weaker than the one stated in Theorem \ref{theo:1}. However, it is enough for the needs of \cite{BDIK}. The full statement of Theorem \ref{theo:1}, that is the extension of \eqref{p:10} to the whole range $0\leq v<s^{\frac{1}{3}}$ with uniform constants $s_0,c_1$ and $c_2$ will be given in the next section with the use of an alternative method based on the recent result \cite{CK}.

\end{rem}
\section{Extending \eqref{e:2} by Toeplitz determinant techniques}\label{sec:3}
Consider the following function on the unit circle,
\begin{equation*}
	f^t(\e^{\im\theta})=\begin{cases}\e^{-2v},&\theta\in[0,t)\cup[2\pi-t,2\pi)\\ 1,&\theta\in[t,2\pi-t)\end{cases};\ \ \ \ 0<t<2\pi,
\end{equation*}
with Fourier coefficients $\{f_k^t\}_{k\in\mathbb{Z}}$ and associated Toeplitz determinant $D_n(f^t)$,
\begin{equation*}
	f_k^t=\frac{1}{2\pi}\int_0^{2\pi}f^t\big(\e^{\im\theta}\big)\e^{-\im k\theta}\d\theta,\hspace{1cm} D_n(f^t)=\det\big[f_{j-k}^t\big]_{j,k=0}^{n-1}.
\end{equation*}
We first observe
\begin{lem}\label{lemT} For any fixed $s>0$,
\begin{equation*}
	\lim_{n\rightarrow\infty}D_n(f^{\frac{2s}{n}})=\det(I-\gamma K_s)\big|_{L^2(-1,1)};\ \ \ \ K_s(\lambda,\mu)=\frac{\sin s(\lambda-\mu)}{\pi(\lambda-\mu)}.
\end{equation*}
\end{lem}
\begin{proof} A straightforward computation of the Fourier coefficents $f_k^{\frac{2s}{n}}$ gives
\begin{equation*}
	f_k^{\frac{2s}{n}}=-\gamma\frac{\sin(k\frac{2s}{n})}{\pi k},\ \ k\neq 0;\ \ \ \ \ \ f_0^{\frac{2s}{n}}=1-\gamma\frac{2s}{n\pi}
\end{equation*}
and therefore
\begin{equation*}
	n f_{j-k}^{\frac{2s}{n}}\rightarrow I(\lambda,\mu)-\gamma K_s(\lambda,\mu),\ \ \ \textnormal{if}\ \ \  \frac{j}{n}\rightarrow\lambda,\ \frac{k}{n}\rightarrow\mu.
\end{equation*}
This, together with the translation invariance of $K_s(\lambda,\mu)$, implies the Lemma by standard properties of trace class operators.
\end{proof}
We will now obtain asymptotics of $D_n(f^t)$ as $n\rightarrow\infty$, for $s>s_0$ with $s_0>0$ sufficiently large and $0\leq v<s^{\frac{1}{3}}$. Note to this end that the function $f^t(z)$ is of Fisher-Hartwig type, see, for example, \cite{DIII}, with two jump-type singularities at $z_1=\e^{\im t},z_2=\e^{\im(2\pi-t)}$ and parameters
\begin{equation*}
	\beta=\beta_1=-\beta_2=\frac{\im v}{\pi}.
\end{equation*}
In more detail, we have
\begin{equation}\label{FishHart}
	f^t(z)=\e^{-\frac{2vt}{\pi}}g_{z_1,\beta_1}(z)g_{z_2,\beta_2}(z)z_1^{-\beta_1}z_2^{-\beta_2},\ \ \ z=\e^{\im\theta}
\end{equation}
where
\begin{equation*}
	g_{z_j,\beta_j}(\e^{\im\theta})=\begin{cases} \e^{\im\pi\beta_j},&\theta\in[0,\textnormal{arg}\,z_j)\\ \e^{-\im\pi\beta_j},&\theta\in[\textnormal{arg}\,z_j,2\pi) \end{cases}.
\end{equation*}
The asymptotics of $D_n(f^t)$ as $n\rightarrow\infty$ for any {\it fixed} $t$ is a classical result going back to the works of Widom, Basor, B\"ottcher and Silbermann \cite{W,B,BS}. As is shown in \cite{CK}, Theorem $1.11$, this result still holds in the case of $\frac{2s}{n}<t<t_0$ for sufficiently small $t_0$ with adjusted error estimate:
\begin{equation}\label{t:0}
	\ln D_n(f^t)=-\frac{4v}{\pi}s+\frac{2v^2}{\pi^2}\ln\left(2n\sin\frac{2s}{n}\right)+2\ln G\left(1+\frac{\im v}{\pi}\right)G\left(1-\frac{\im v}{\pi}\right)+r(s,v,n),
\end{equation}
and there exists $n_0(v),s_0(v),C_0(v)>0$ such that\footnote{note that $\beta_j$ in \cite{CK}, and hence $v$, are fixed.}
\begin{equation*}
	\forall\, n>n_0,\ s>s_0,\ v\geq 0\ \ \textnormal{fixed}:\ \  \ \ \big|r(s,v,n)\big|<C_0(v)s^{-1}.
\end{equation*}
We will now extend the argument of \cite{CK} to the case of $0\leq v<s^{\frac{1}{3}}$ and see that the error term remains small (for large $s$) in this region of the $(v,s)$-plane. In order to carry out this approach we have to track the dependence of the error term \cite{CK}$(7.58)$, the error term \cite{CK}$(8.20)$ at $z_1$ and a similar term at $z_2$ on $\beta=\frac{\im v}{\pi}$.\smallskip

First, we obtain by a straightforward calculation that the crucial constants $\widehat{E}_j(z_j)\e^{-\im\pi\beta_j\sigma_3}$ in \cite{CK}$(7.47),(7.52)$ are bounded in $n,t>0$ uniformly for any purely imaginary $\beta_j$ and $\frac{2s}{n}<t<t_0$.\smallskip

Second, we consider the auxiliary $M$-RHP of Section $4$ in \cite{CK}:
\begin{problem}[see \cite{CK}, Section 4, or \cite{CIK}, Section $4.2.1$] Determine $M(\z)=M(\z;\beta)\in\mathbb{C}^{2\times 2}$ such that
\begin{enumerate}
	\item $M(\z)$ is analytic for $\z\in\mathbb{C}\backslash\big(\e^{\pm\im\frac{\pi}{4}}\mathbb{R}\cup[0,\infty)\big)$. We orient the five jump rays as shown in Figure \ref{figure4} below.
	\item The boundary values $M_{\pm}(\z)$ on $\e^{\pm\im\frac{\pi}{4}}\mathbb{R}\cup(0,\infty)$ are continuous and related by the following jump conditions:
	\begin{equation*}
		M_+(\z)=M_-(\z)\bigl(\begin{smallmatrix} 1 & \e^{-\im\pi\beta}\\ 0 & 1\end{smallmatrix}\bigr),\ \ \z\in\e^{\im\frac{\pi}{4}}(0,\infty);\ \ \ \ M_+(\z)=M_-(\z)\bigl(\begin{smallmatrix} 1 & 0\\ -\e^{\im\pi\beta} & 1\end{smallmatrix}\bigr),\ \ \z\in\e^{\im\frac{3\pi}{4}}(0,\infty)
	\end{equation*}
	\begin{equation*}
		M_+(\z)=M_-(\z)\bigl(\begin{smallmatrix} 1 & 0\\ \e^{-\im\pi\beta} & 1\end{smallmatrix}\bigr),\ \ \z\in\e^{\im\frac{5\pi}{4}}(0,\infty);\ \ \ \ M_+(\z)=M_-(\z)\bigl(\begin{smallmatrix} 1 & -\e^{\im\pi\beta}\\ 0 & 1\end{smallmatrix}\bigr),\ \ \z\in\e^{\im\frac{7\pi}{4}}(0,\infty)
	\end{equation*}
	and
	\begin{equation*}
		M_+(\z)=M_-(\z)\e^{2\pi\im\beta\sigma_3},\ \ \ \z\in(0,\infty).
	\end{equation*}
	\item As $\z\rightarrow\infty$, valid in a full vicinity of $\z=\infty$,
	\begin{equation*}
		M(\z)=\Big(I+\frac{M_1}{\z}+\mathcal{O}\left(\z^{-2}\right)\Big)\z^{-\beta\sigma_3}\e^{-\frac{1}{2}\z\sigma_3},\ \ \ \ \textnormal{arg}\,\z\in(0,2\pi).
	\end{equation*}
	with
	\begin{equation*}
		M_1=\begin{pmatrix}-\beta^2& -\e^{-2\pi\im\beta}\frac{\Gamma(1-\beta)}{\Gamma(\beta)}\\ \e^{2\pi\im\beta}\frac{\Gamma(1+\beta)}{\Gamma(-\beta)} & \beta^2\end{pmatrix}.
	\end{equation*}
\end{enumerate}
\end{problem}
\begin{figure}[tbh]
\begin{center}
\resizebox{0.4\textwidth}{!}{\includegraphics{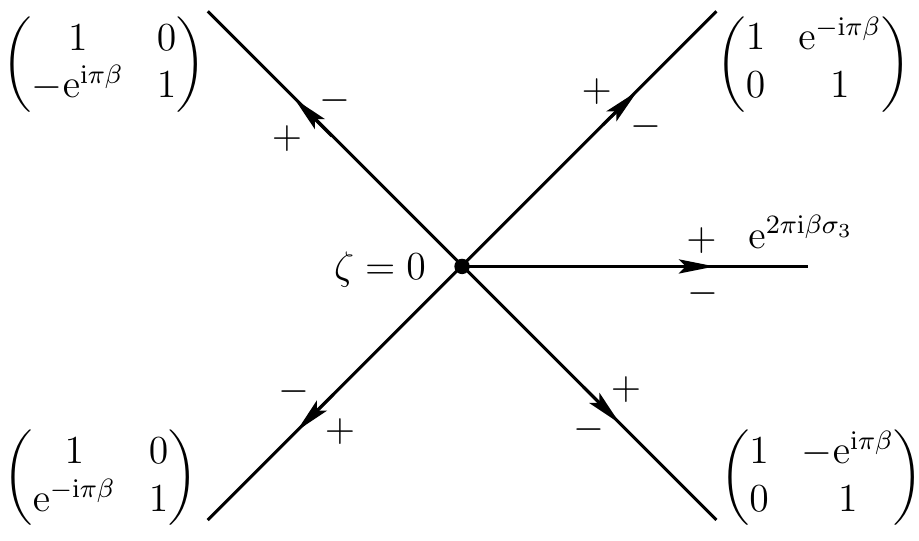}}
\caption{Jump behavior of $M(\z)$ in the complex $\z$-plane.}
\label{figure4}
\end{center}
\end{figure}
This problem is solved explicitly in terms of the confluent hypergeometric function $U(a;\z)=U(a,1;\z)$, compare RHP \ref{bareconf} and \eqref{p:2} above. We define
\begin{equation*}
	M(z)=\begin{pmatrix} U(\beta,z) & U(1-\beta,\e^{-\im\pi}z)\e^{-\im\pi\beta}\frac{\Gamma(1-\beta)}{\Gamma(\beta)} \\ U(1+\beta,z)\e^{2\pi\im\beta}\frac{\Gamma(1+\beta)}{\Gamma(-\beta)} & U(-\beta,\e^{-\im\pi}z)\e^{\im\pi\beta}\end{pmatrix}\e^{-\frac{1}{2}z\sigma_3}\mathcal{M}(z),\ \ \ \textnormal{arg}\,z\in(0,2\pi)
\end{equation*}
with
\begin{equation*}
	\mathcal{M}(z)=\begin{cases}\bigl(\begin{smallmatrix}1&-\e^{-\im\pi\beta}\\ 0 & 1\end{smallmatrix}\bigr),&\textnormal{arg}\,z\in(0,\frac{\pi}{4})\\ I,&\textnormal{arg}\,z\in(\frac{\pi}{4},\frac{3\pi}{4})\\ \bigl(\begin{smallmatrix} 1 & 0\\ -\e^{\im\pi\beta} & 1\end{smallmatrix}\bigr),&\textnormal{arg}\,z\in(\frac{3\pi}{4},\frac{5\pi}{4})\end{cases},\ \ \ \mathcal{M}(z)=\begin{cases} \bigl(\begin{smallmatrix} 1 & 0\\ -2\im\sin\pi\beta & 1\end{smallmatrix}\bigr),&\textnormal{arg}\,z\in(\frac{5\pi}{4},\frac{7\pi}{4})\\ \bigl(\begin{smallmatrix} 1 & -\e^{\im\pi\beta}\\ -2\im\sin\pi\beta & \e^{2\pi\im\beta}\end{smallmatrix}\bigr),&\textnormal{arg}\,z\in(\frac{7\pi}{4},2\pi)\end{cases}
\end{equation*}
From the standard asymptotics of the confluent hypergeometric function (cf. \cite{NIST}), we conclude that for purely imaginary $\beta$,
\begin{equation}\label{t:1}
	M(z)=\left[I+\e^{-\im\pi\beta\sigma_3}\left\{\mathcal{O}\left(\frac{\beta}{z}\right)+\mathcal{O}\left(\frac{\beta^2}{z}\right)\right\}\e^{-\im\pi\beta\sigma_3}\right]z^{-\beta\sigma_3}\e^{-\frac{z}{2}\sigma_3},\ \ \ z\rightarrow\infty.
\end{equation}
Here we used the identity (cf. \cite{NIST})
\begin{equation*}
	\left|\frac{\Gamma(1-\beta)}{\Gamma(\beta)}\right|=|\beta|,\ \ \ \beta\in\im\mathbb{R}.
\end{equation*}
Now we are ready to estimate the jumps in the final $\Upsilon$-RHP in the region $\frac{2s}{n}<t<t_0$, see \cite{CK}$(7.55)$. The function $\Upsilon(\lambda)$ is the final ratio function for the underlying nonlinear steepest descent analysis in \cite{CK}. We will not explicitly define $\Upsilon(z)$ at this point as it would involve several other model functions which we do not need for our purposes. Instead we refer to \cite{CK},(7.55), (7.56) and (7.57) for details. The jump contour $\Sigma_{\Upsilon}$ in the $\Upsilon$-RHP consists of two contracting circles of radii $\mathcal{O}(t)$ and the lens boundary around the larger arc outside the circles. The jump on the outer lens lip equals identity plus a constant matrix multiplied by $z^{-n}(\mathcal{D}_{\textnormal{in},t}(z)\mathcal{D}_{\textnormal{out},t}(z))^{-1}$, defined in \cite{CK}$(7.3),(7.4)$. Subject to the condition that, say, $v^2<s,n>s>s_0,\frac{2s}{n}<t<t_0$, we obtain the estimate
\begin{equation*}
	\Big|z^{-n}\big(\mathcal{D}_{\textnormal{in},t}(z)\mathcal{D}_{\textnormal{out},t}(z)\big)^{-1}\Big|<c_0\e^{-\epsilon nt}\left|\left(\frac{z-z_1}{z-z_2}\right)^{-2\beta}\left(\frac{z_1}{z_2}\right)^{\beta}\right|<c_0\e^{-\epsilon nt+C_0v}<c_0\e^{-\epsilon_1nt}
\end{equation*}
for some $c_0,C_0,\epsilon,\epsilon_1>0$ on the outer lip of the lens. Thus the jump is exponentially (in $s$) close to the identity. The corresponding estimate on the inner lens lip is similar. For the jump on the circle centered at $z_1$ we use \eqref{t:1} and the properties of $\widehat{E}_1(z)$. This leads us to an estimate of the form
\begin{equation*}
	I+\mathcal{O}\left(\frac{\beta}{nt}\right)+\mathcal{O}\left(\frac{\beta^2}{nt}\right)
\end{equation*}
and a similar one also holds on the circle centered at $z_2$. From standard theory \cite{DZ} it follows now that the problem for $\Upsilon(z)$ is solvable and we obtain the estimates (which replace \cite{CK}$(7.58)$),
\begin{equation}\label{t:2}
	\Upsilon(z)=I+\mathcal{O}\left(\frac{\beta}{nt}\right)+\mathcal{O}\left(\frac{\beta^2}{nt}\right),\ \ \ \frac{\d}{\d z}\Upsilon(z)=\frac{1}{t}\Big\{\mathcal{O}\left(\frac{\beta}{nt}\right)+\mathcal{O}\left(\frac{\beta^2}{nt}\right)\Big\}
\end{equation}
as $s,n\rightarrow\infty,n>s>s_0,v^2<s$, uniformly for $z$ off the jump contour $\Sigma_{\Upsilon}$ and uniformly in $s,v$ such that $\frac{2s}{n}<t<t_0$. Using \eqref{t:2} we can then estimate the error term in \cite{CK}$(8.20)$.
\begin{rem} For simplicity of the derivation, it is assumed in \cite{CK}$(8.20)$ that there also exist $\alpha>0$ Fisher-Hartwig singularities at the same points: 
\begin{equation*}
	f^t(z)=\e^{-\frac{2vt}{\pi}}\prod_{j=1}^2|z-z_j|^{2\alpha}g_{z_j,\beta_j}(z)z_j^{-\beta_j},
\end{equation*}
as opposed to \eqref{FishHart}. In the end of the derivation, to obtain the actual error term in the differential identity for the determinant $D_n(f^t)$, we have to multiply \cite{CK}$(8.20)$ by $\alpha$ and take the limit $\alpha\rightarrow 0$.
\end{rem}
By a straightforward computation we obtain
\begin{equation*}
	\frac{1}{t}\Big\{\mathcal{O}\left(\frac{\beta}{nt}\right)+\mathcal{O}\left(\frac{\beta^3}{nt}\right)\Big\}
\end{equation*}
as error term in the differential identity for $D_n(f^t)$, uniformly for $\frac{2s}{n}<t<t_0$, the rest of the identity for $D_n(f^t)$ is the same as in \cite{CK}. Integrating this identity over $\frac{2s}{n}<t<t_0$ for some fixed $t_0$ and using the well-known large $n$ asymptotics for the determinant with two fixed Fisher-Hartwig singularities at $z_1=\e^{\im t_0},z_2=\e^{\im(2\pi-t_0)}$ (see, for instance, \cite{CK}$(1.8)$), we derive \eqref{t:0}, where now, for some $s_0,C_1,C_2>0$,
\begin{equation*}
	n>s>s_0,\ \ \ \ \ 0\leq v<s^{\frac{1}{3}},\ \ \ \ \ \big|r(s,v,n)\big|<C_1\frac{v}{s}+C_2\frac{v^3}{s}.
\end{equation*}
At this point we take the limit $n\rightarrow\infty$, use Lemma \ref{lemT} and obtain Theorem \ref{theo:1}.
\section{More on the integral term -- proof of Proposition \ref{prop:int}}\label{sec:4}
We first analyze all objects involved in the formul\ae\, for $M(x,\kappa)$ in the limit $\kappa\downarrow 0$. Observe that the defining equation \eqref{e:0} for the branch point $a=a(\kappa)$ can be written as
\begin{equation}\label{m:1}
	F(a,\kappa)\equiv\kappa+\frac{\im}{2}\oint_{\Sigma}\left(\frac{\mu^2-a^2}{\mu^2-1}\right)^{\frac{1}{2}}\d\mu=0
\end{equation}
where $\Sigma\subset\mathbb{C}$ is a simple counterclockwise oriented Jordan curve around the interval $[a,1]$ and we fix
\begin{equation*}
	-\pi<\textnormal{arg}\left(\frac{\mu^2-a^2}{\mu^2-1}\right)\leq\pi\ \ \ \textnormal{with}\ \ \sqrt{\frac{\mu^2-a^2}{\mu^2-1}}>0\ \ \textnormal{for}\ \ \mu>1.
\end{equation*} 
Since $F(a,\kappa)$ is analytic at $(a,\kappa)=(1,0)$ and $F_a(1,0)=\frac{\pi}{2}>0$ we have unique solvability of equation \eqref{m:1} in a neighborhood of $(a,\kappa)=(1,0)$ and the solution $a=a(\kappa)$ is analytic at $\kappa=0$, i.e. \cite{BDIK}$(2.5)$ extends to a full Taylor series
\begin{equation}\label{m:2}
	a(\kappa)=1-\frac{2\kappa}{\pi}-\frac{\kappa^2}{\pi^2}+\sum_{j=3}^{\infty}a_j\kappa^j,\ \ \ \kappa\rightarrow 0.
\end{equation}
Next we extend the small $\kappa$ expansions listed in \cite{BDIK}$(2.10),(2.11)$ for the frequency $V=V(\kappa)$ and nome $\tau=\tau(\kappa)$. Both, compare \eqref{e:-1}, are expressed in terms of complete elliptic integrals and thus hypergeometric functions, cf. \cite{NIST},
\begin{equation*}
	K(a)=\frac{\pi}{2}\, {_2}F_1\left(\frac{1}{2},\frac{1}{2};1;a^2\right),\ \ \ \ \ E(a)=\frac{\pi}{2}\,{_2}F_1\left(-\frac{1}{2},\frac{1}{2};1;a^2\right).
\end{equation*}
Hence, using expansion of the hypergeometric functions at unity (cf. \cite{NIST}) and combining these with \eqref{m:2} we obtain
\begin{equation}\label{m:3}
	V(\kappa)=-\frac{2}{\pi}\left(1+\ln\kappa\sum_{j=1}^{\infty}u_j\kappa^j+\sum_{j=1}^{\infty}v_j\kappa^j\right),\ \ \kappa\downarrow 0;\ \ \ \ \ u_1=\frac{1}{\pi},\ \ v_1=-\frac{1}{\pi}(1+\ln 4\pi);
\end{equation}
as well as
\begin{equation}\label{m:4}
	\tau(\kappa)=-\frac{2\im}{\pi}\ln\kappa+\ln\kappa\sum_{j=1}^{\infty}b_j\kappa^j+\sum_{j=0}^{\infty}c_j\kappa^j,\ \ \kappa\downarrow 0;\ \ \ \ \ c_0=\frac{2\im}{\pi}\ln 4\pi.
\end{equation}
The series in the right hand sides of \eqref{m:3} and \eqref{m:4} are convergent series.
\begin{rem}\label{normconst} Another parameter which appears in \cite{BDIK}$(1.16)$ and which serves as normalization for the elliptic nome $\tau=\tau(\kappa)$ is given by
\begin{equation*}
	c=c(\kappa)=\frac{\im}{2}\frac{1}{K(a')},\ \ \ \ a'=\sqrt{1-a^2}.
\end{equation*}
Hence analyticity of $a(\kappa)$ at $\kappa=0$ implies at once analyticity of $c(\kappa)$ at the same point and \cite{BDIK}, Corollary $2.2$ extends to a full Taylor series
\begin{equation*}
	c(\kappa)=\frac{\im}{\pi}\left(1-\frac{\kappa}{\pi}+\sum_{j=2}^{\infty}d_j\kappa^j\right),\ \ \ \kappa\rightarrow 0.
\end{equation*}
\end{rem}
At this point we recall \cite{BDIK}, Proposition $4.2$: It was shown that $M(x,\kappa)$, given by \cite{BDIK}$(1.24)$ (see also \eqref{Mid}) and defined for $x\in\mathbb{R},\kappa\in(0,1-\delta]$ is smooth in both its arguments and one-periodic in the first,
\begin{equation*}
	M(x+1,\kappa)=M(x,\kappa).
\end{equation*}
Hence we can expand $M(x,\kappa)$ in a Fourier series
\begin{equation}\label{m:5}
	M(x,\kappa)=\sum_{n\in\mathbb{Z}}a_n(\kappa)\e^{2\pi\im nx};\ \ a_n(\kappa)=\frac{1}{2\pi}\int_0^1M(x,\kappa)\e^{-2\pi\im nx}\d x
\end{equation}
and the series converges uniformly in $x\in\mathbb{R},\kappa\in(0,1-\delta]$ to $M(x,\kappa)$. Note that for $n\neq 0$,
\begin{equation*}
	a_n(\kappa)=-\frac{1}{8\pi^2n^2}\int_0^1\left[\frac{\partial^2}{\partial x^2}M(x,\kappa)\right]\e^{-2\pi\im nx}\d x;\ \ \ \frac{\d}{\d\kappa}a_n(\kappa)=-\frac{1}{8\pi^3n^2}\int_0^1\left[\frac{\partial^3}{\partial x^2\partial\kappa}M(x,\kappa)\right]\e^{-2\pi\im nx}\d x.
\end{equation*}
In the derivation of Proposition \ref{prop:int} we will use the series representation \eqref{m:5}, i.e. we need estimates for $a_n(\kappa)$ and its derivative.\smallskip

The building blocks of $M(x,\kappa)$ are summarized in \cite{BDIK}$(4.10),(4.11),(4.12)$ (see also \eqref{Mid}, \eqref{Mid2} and \eqref{Mid3}) and these are all functions involving the third Jacobi theta function
\begin{equation*}
	\theta_3\big(x|\tau(\kappa)\big)=\sum_{m\in\mathbb{Z}}\e^{\im\pi\tau m^2+2\pi\im mx}.
\end{equation*}
The roots of this function are located at $\frac{1}{2}+\frac{\tau}{2}+\mathbb{Z}+\tau\mathbb{Z}$ and from \eqref{m:4} we see that uniformly in $x\in\mathbb{R}$,
\begin{equation*}
	\theta_3(x|\tau)=1+\mathcal{O}\left(\kappa^2\right).
\end{equation*}
Hence, for any $m,\ell\in\mathbb{Z}_{\geq 0}$ there exist $C_{m,\ell},C_m>0$\footnote{In what follows, $C$ with or without indices denotes a positive constant independent of $x$ and $\kappa$ whose value may be different in different estimates.} such that
\begin{equation}\label{m:6}
	\left|\frac{\partial^{m+\ell}}{\partial x^m\partial\tau^{\ell}}\theta_3(x|\tau)-\delta_{m0}\delta_{\ell0}\right|\leq C_{m,\ell}\kappa^2;\ \ \ \ \ \ \ \ \ \left|\frac{\partial^{m+1}}{\partial x^m\partial\kappa}\theta_3\big(x|\tau(\kappa)\big)\right|\leq C_m\kappa,
\end{equation}
for all $x\in\mathbb{R},\kappa\in(0,1-\delta]$ and in particular
\begin{equation}\label{m:7}
	\left|\frac{\partial^m}{\partial x^m}\left[\frac{\im}{4\pi}\frac{\kappa}{\theta_3(x|\tau)}\frac{\partial^2}{\partial y^2}\theta_3(y|\tau)\Big|_{y=x}\frac{\d\tau}{\d\kappa}\right]\right|\leq C_m\kappa^2,\ \ \ \ \ \left|\frac{\partial^{m+1}}{\partial x^m\partial\kappa}\left[\frac{\im}{4\pi}\frac{\kappa}{\theta_3(x|\tau)}\frac{\partial^2}{\partial y^2}\theta_3(y|\tau)\Big|_{y=x}\frac{\d\tau}{\d\kappa}\right]\right|\leq C_m\kappa.
\end{equation}
The last two estimates are needed for the second summand in \cite{BDIK}$(1.24)$ (compare \eqref{Mid}). For the first term in loc. cit. we collect from \eqref{m:2}, the uniform bounds which are valid for all $x\in\mathbb{R},\kappa\in(0,1-\delta]$
\begin{equation}\label{m:8}
	C_{1,m}\leq\left|\frac{\d^m}{\d\kappa^m}a(\kappa)\right|\leq C_{2,m},\ \ \ \ \ C_{1,m}\leq\left|\frac{\d^m}{\d\kappa^m}\frac{1}{a(\kappa)(1+a(\kappa))}\right|\leq C_{2,m}.
\end{equation}
Next we analyze $\Xi_0(x,\kappa),\Xi_2(x,\kappa)$ and $\Theta_0(x,\kappa)$: The functions $\Xi_j(x,\kappa)$, see \cite{BDIK}$(4.10)$ and \eqref{Mid2}, are combinations of $\theta_j(x|\tau)$ and thus, following the logic which lead to \eqref{m:6}, we obtain similarly
\begin{equation*}
	\left|\frac{\partial^{m+\ell}}{\partial x^m\partial\tau^{\ell}}\theta_0(x\pm d|\tau)-\delta_{m0}\delta_{\ell 0}\right|\leq C_{m,\ell}\kappa;\ \ \ \ \ \ \left|\frac{\partial^{m+1}}{\partial x^m\partial\kappa}\theta_0(x\pm d|\tau)\right|\leq C_m
\end{equation*}
for all $x\in\mathbb{R},\kappa\in(0,1-\delta]$. The same estimates also hold for $\theta_2(x\pm d|\tau)$ which yield
\begin{equation}\label{m:9}
	\big|\Xi_j(x,\kappa)\big|\leq C,\ \ \ \ \left|\frac{\partial^m}{\partial x^m}\Xi_j(x,\kappa)\right|\leq C_m\kappa,\ \ \ \ \ \left|\frac{\partial^{m+1}}{\partial x^m\partial\kappa}\Xi_j(x,\kappa)\right|\leq C_m,\ \ \ \ j=0,2,
\end{equation}
valid for all $x\in\mathbb{R},\kappa\in(0,1-\delta]$. Since the remaining function $\Theta_0(x,\kappa)$ satisfies estimates of the same type as \eqref{m:9}, we can combine \eqref{m:7}, \eqref{m:8} and \eqref{m:9} to derive
\begin{prop}\label{p42} Given $m\in\mathbb{Z}_{\geq 1}$ there exist positive $C,C_m$ such that $M(x,\kappa)$ given in \cite{BDIK}$(1.24)$ satisfies
\begin{equation*}
	\big|M(x,\kappa)\big|\leq C,\ \ \ \ \ \left|\frac{\partial^m}{\partial x^m}M(x,\kappa)\right|\leq C_m\kappa,\ \ \ \ \ \ \ \left|\frac{\partial^{m+1}}{\partial x^m\partial\kappa}M(x,\kappa)\right|\leq C_m
\end{equation*}
for all $x\in\mathbb{R},\kappa\in(0,1-\delta]$.
\end{prop}
This Proposition allows us to estimate the Fourier coefficients in \eqref{m:5}
\begin{cor}\label{Fcoeff} For all $n\in\mathbb{Z},\kappa\in(0,1-\delta]$,
\begin{equation*}
	\big|a_n(\kappa)\big|\leq\frac{C}{1+n^2},\ \ \ \ \ \left|\frac{\d}{\d\kappa}a_n(\kappa)\right|\leq\frac{C}{1+n^2},\ \ \ \ C>0.
\end{equation*}
\end{cor}
In order to complete the proof of Proposition \ref{prop:int} we use \eqref{m:5}, \cite{BDIK}(4.4), Proposition \ref{p42} and integration by parts
\begin{align*}
	\int_s^{\infty}M&\Big(tV\big(vt^{-1}\big),vt^{-1}\Big)\frac{\d t}{t}=\int_0^{\kappa}a_0(u)\frac{\d u}{u}+\sum_{n\neq 0}\int_s^{\infty}a_n(vt^{-1})\exp\Big[-2\pi\im ntV\big(vt^{-1}\big)\Big]\frac{\d t}{t}\\
	&\,=\int_0^{\kappa}a_0(u)\frac{\d u}{u}-\frac{1}{4\pi sc(\kappa)}\sum_{n\neq 0}\frac{a_n(\kappa)}{n}\e^{-2\pi\im nsV(\kappa)}-\frac{1}{4\pi}\sum_{n\neq 0}\frac{1}{n}\int_s^{\infty}\frac{\partial}{\partial t}\left[\frac{a_n(vt^{-1})}{tc(vt^{-1})}\right]\e^{-2\pi\im ntV(vt^{-1})}\d t,
\end{align*}
We estimate the infinite sums as follows: For the first sum we use, compare Remark \ref{normconst}, the fact that
\begin{equation*}
	\left|\frac{1}{c(\kappa)}\right|\leq C,\ \ \ \forall\, \kappa\in(0,1-\delta],
\end{equation*}
together with Corollary \ref{Fcoeff}. Thus 
\begin{equation}\label{m:10}
	\left|\frac{1}{sc(\kappa)}\sum_{n\neq 0}\frac{a_n(\kappa)}{n}\e^{-2\pi\im nsV(\kappa)}\right|\leq\frac{C}{s},\ \ \kappa=\frac{v}{s}\in(0,1-\delta].
\end{equation}
For the second sum we differentiate,
\begin{equation*}
	\frac{\partial}{\partial t}\left[\frac{a_n(vt^{-1})}{tc(vt^{-1})}\right]=-\frac{\kappa}{t^2c(\kappa)}\frac{\d}{\d\kappa}a_n(\kappa)-\frac{a_n(\kappa)}{(tc(\kappa))^2}\left(c(\kappa)-\kappa\frac{\d}{\d\kappa}c(\kappa)\right),\ \ \kappa=\frac{v}{t}
\end{equation*}
and since in the integral $0<\frac{v}{t}\leq\frac{v}{s}\leq 1-\delta$ we obtain using Corollary \ref{Fcoeff} and Remark \ref{normconst} the bound
\begin{equation*}
	\left|\frac{\partial}{\partial t}\left[\frac{a_n(vt^{-1})}{tc(vt^{-1})}\right]\right|\leq\frac{C}{t^2(1+n^2)}.
\end{equation*}
Thus,
\begin{equation*}
	\left|\sum_{n\neq 0}\frac{1}{n}\int_s^{\infty}\frac{\partial}{\partial t}\left[\frac{a_n(vt^{-1})}{tc(vt^{-1})}\right]\e^{-2\pi\im ntV(vt^{-1})}\d t\right|\leq\frac{C}{s},
\end{equation*}
which, together with \eqref{m:10}, completes the proof of Proposition \ref{prop:int}, i.e.
\begin{equation*}
	\int_s^{\infty}M\Big(tV\big(vt^{-1}\big),vt^{-1}\Big)\frac{\d t}{t}=\int_0^{\kappa}a_0(u)\frac{\d u}{u}+\mathcal{O}\left(s^{-1}\right)
\end{equation*}
uniformly for $s\geq s_0>0$ and $0<v\leq s(1-\delta)$.
\begin{appendix}
\section{Explicit form of $M(z,\kappa)$}\label{appA} We fix $a=a(\kappa),\tau=\tau(\kappa)$ and $V=V(\kappa)$ as in \eqref{e:0}, \eqref{e:-1} throughout. Note that in terms of $\theta_3(z|\tau)$ defined in \eqref{e:-1} the remaining three Jacobi theta functions $\theta_0(z),\theta_1(z),\theta_2(z)$ are given by (cf. \cite{NIST})
\begin{equation*}
	\theta_0(z|\tau)=\theta_3\left(z+\frac{1}{2}\bigg|\tau\right),\ \ \ \theta_2(z|\tau)=\e^{\im\frac{\pi}{4}\tau+\im\pi z}\theta_3\left(z+\frac{\tau}{2}\bigg|\tau\right),\ \ \ \theta_1(z|\tau)=-\im\e^{\im\frac{\pi}{4}\tau+\im\pi z}\theta_3\left(z+\frac{1}{2}+\frac{\tau}{2}\bigg|\tau\right).
\end{equation*}
The integral term $M(x,\kappa)$, defined in \cite{BDIK},$(1.24)$, equals
\begin{equation}\label{Mid}
	M(x,\kappa)=\frac{\Xi_0(x,\kappa)\Theta_0(x,\kappa)+6a(\kappa)\Xi_2(x,\kappa)}{48a(\kappa)(1+a(\kappa))}+\frac{\im}{4\pi}\frac{\kappa}{\theta_3(x|\tau)}\frac{\partial^2}{\partial y^2}\theta_3(y|\tau)\Big|_{y=x}\frac{\d\tau}{\d\kappa}
\end{equation}
with $\Xi_j(x,\kappa)$ and $\Theta_0(x,\kappa)$ equal to the following functions:
\begin{equation}\label{Mid2}
	\Xi_j(x,\kappa)=2\frac{\theta_3^2(0|\tau)}{\theta_3^2(x|\tau)}\frac{\theta_j(x+d|\tau)\theta_j(x-d|\tau)}{\theta_j^2(d|\tau)},\ \ j=0,2;\hspace{1cm} d=-\frac{\tau}{4},
\end{equation}
and
\begin{align}
	\Theta_0(x,\kappa)=&\,5c^2\left\{\frac{\theta_0''(x+d|\tau)}{\theta_0(x+d|\tau)}-2\left(\frac{\theta_0'(d|\tau)}{\theta_0(d|\tau)}\right)'+\frac{\theta_0''(x-d|\tau)}{\theta_0(x-d|\tau)}\right\}+14c^2\frac{\theta_0'(x-d|\tau)}{\theta_0(x-d|\tau)}\frac{\theta_0'(x+d|\tau)}{\theta_0(x+d|\tau)}\nonumber\\
	&\,-4c^2\left\{\frac{\theta_0'(x-d|\tau)}{\theta_0(x-d|\tau)}+\frac{\theta_0'(d|\tau)}{\theta_0(d|\tau)}-\frac{\theta_0'(x+d|\tau)}{\theta_0(x+d|\tau)}\right\}\frac{\theta_0'(d|\tau)}{\theta_0(d|\tau)}-2c(1+a)\bigg\{\frac{\theta_0'(x+d|\tau)}{\theta_0(x+d|\tau)}-2\frac{\theta_0'(d|\tau)}{\theta_0(d|\tau)}\nonumber\\
	&\,-\frac{\theta_0'(x-d|\tau)}{\theta_0(x-d|\tau)}\bigg\}-2(2+a);\hspace{1cm}\ \ \ \ \ \ \ \theta'(x|\tau)=\frac{\d}{\d x}\theta(x|\tau).\label{Mid3}
\end{align}
Here, compare Remark \ref{normconst},
\begin{equation*}
	c=c(\kappa)=\frac{\im}{2}\frac{1}{K(a')},\ \ a'=\sqrt{1-a^2}.
\end{equation*}
\end{appendix}


\end{document}